\documentclass[journal]{IEEEtran}

\usepackage{cite}

\ifCLASSINFOpdf
   \usepackage[pdftex]{graphicx}
   \graphicspath{{../pdf/}{../jpeg/}}
   \DeclareGraphicsExtensions{.pdf,.jpeg,.png}
\else
   \usepackage[dvips]{graphicx}
   \graphicspath{{../eps/}}
   \DeclareGraphicsExtensions{.eps}
\fi
\usepackage{subcaption}
\usepackage{tikz}

\usepackage{amsmath}

\usepackage{array}

\usepackage{url}

\DeclareMathAlphabet{\mathcal}{OMS}{cmsy}{m}{n} %
\usepackage{algorithm}%
\usepackage{algpseudocode}%
\usepackage{siunitx}
\usepackage{import} %
\usepackage{xcolor}
\usepackage{amssymb}  %
\usepackage{amsthm} %
\usepackage{multirow}
\usepackage{multicol}

\newtheorem{assumption}{Assumption}

\newtheorem{proposition}{Proposition}
\newtheorem{theorem}{Theorem}

\newtheorem{definition}{Definition}
\newtheorem{corollary}{Corollary}
\newtheorem{remark}{Remark}
\usepackage{bm}
\usepackage{hyperref}

\def\vzero{{0}}

\def\vc{{c}}
\def\vd{{d}}

\def\vf{{f}}
\def\vg{{g}}
\def\vh{{h}}

\def\vp{{p}}

\def\vu{{u}}
\def\vv{{v}}

\def\vx{{x}}

\def\evc{{c}}
\def\evd{{d}}

\def\evf{{f}}

\def\mA{{A}}
\def\mB{{B}}

\def\mE{{E}}

\def\mQ{{Q}}
\def\mR{{R}}

\DeclareMathAlphabet{\mathsfit}{\encodingdefault}{\sfdefault}{m}{sl}
\SetMathAlphabet{\mathsfit}{bold}{\encodingdefault}{\sfdefault}{bx}{n}

\def\sP{{\mathbb{P}}}

\def\sR{{\mathbb{R}}}
\def\sS{{\mathbb{S}}}

\def\sU{{\mathbb{U}}}

\def\sX{{\mathbb{X}}}

\newcommand\copyrighttext{%
	\footnotesize This work has been received on 24 January 2025 and has been accepted by the IEEE for publication with the Digital Object Identifier:  10.1109/TAC.2025.3558137 on 29 March 2025.
\textcopyright 2025 The Authors. This work is licensed under a Creative Commons Attribution 4.0 License. For more information, see http://creativecommons.org/licenses/by/4.0/
}

\newcommand\copyrightnotice{%
	\begin{tikzpicture}[remember picture,overlay]
		\node[anchor=south,yshift=10pt] at (current page.south) {\fbox{\parbox{\dimexpr\textwidth-\fboxsep-\fboxrule\relax}{\copyrighttext}}};
	\end{tikzpicture}%
}

\begin{document}
\title{Robust Model Predictive Control Exploiting Monotonicity Properties}
\author{Moritz Heinlein$^{1}$, Sankaranarayanan Subramanian$^{2}$
        and Sergio Lucia$^{1}$%
\thanks{*The research leading to these results has received funding from the Deutsche Forschungsgemeinschaft (DFG, German Research Foundation) under grant agreement number 423857295. The authors thank M. Diehl, F. Messerer, M. Molnar and J. Köhler for fruitful discussions.}%
\thanks{$^{1}$ M. Heinlein and S. Lucia are with the Chair of Process Automation Systems at Technische Universität Dortmund, Emil-Figge-Str. 70, 44227 Dortmund, Germany (e-mail: {moritz.heinlein}, {sergio.lucia} @tu-dortmund.de).}%
\thanks{$^{2}$ S. Subramanian is with the Chair of Process Dynamics and Operations, Technische Universität Dortmund, Emil-Figge-Str. 70, 44227 Dortmund, Germany (e-mail: sankaranarayanan.subramanian@tu-dortmund.de).}%
}
\markboth{Accepted for Publication in IEEE Transactions on Automatic Control}%
{Robust model predictive control exploiting monotonicity properties}
\maketitle %
\copyrightnotice %
\begin{abstract}
Robust model predictive control algorithms are essential for addressing unavoidable errors due to the uncertainty in predicting real-world systems. However, the formulation of such algorithms typically results in a trade-off between conservatism and computational complexity.
Monotone systems facilitate the efficient computation of reachable sets and thus the straightforward formulation of a robust model predictive control approach optimizing over open-loop predictions.

We present an approach based on the division of reachable sets to incorporate feedback in the predictions, resulting in less conservative strategies. The concept of mixed-monotonicity enables an extension of our methodology to non-monotone systems. 
Lastly, we discuss the relation between widely used tube-based MPC approaches and our proposed methods.
The potential of the proposed approaches is demonstrated through a nonlinear high-dimensional chemical tank reactor cascade case study.
\end{abstract}

\vspace{-2ex}

\IEEEpeerreviewmaketitle

\section{Introduction}

\IEEEPARstart{M}{odel} predictive control (MPC) has become a widely used control scheme, as it can be applied to multi-input multi-output nonlinear systems with constraints and can control the system optimally according to a specifiable performance metric over a prediction horizon~\cite{rawlingsModelPredictiveControl2017}.
However, the model of the system as well as the future disturbances need to be known accurately in order to guarantee constraint satisfaction, recursive feasibility and stability. As this is a very important requirement, multiple approaches have been proposed to ensure feasibility and stability in the presence of uncertainties.

In contrast to nominal MPC, where the optimal policy is one trajectory of inputs that ignores the presence of uncertainty, the optimal policy for robust MPC is a function of the not-yet-known realizations of the uncertainty, because future actions can be chosen to counteract the realized uncertainty. This dependency is also called recourse.
Open-loop robust approaches do not take recourse into account, as they optimize, similar to nominal MPC, over a single input trajectory which needs to satisfy the constraints for all uncertainties and optimizes typically the worst case of the performance metric over the open-loop predictions~\cite{campoRobustModelPredictive1987}. 
Closed-loop robust approaches optimize over general policies, which is not tractable in general and therefore most closed-loop robust approaches restrict themselves to optimizing over linear policies~\cite{mayneRobustModelPredictive2005}, as also done in tube-based methods~\cite{rawlingsModelPredictiveControl2017}.%

In addition to the consideration of recourse, another challenge for robust approaches is the need to ensure that constraints are satisfied for every possible realization of the uncertainty. This is typically done by computing or approximating the reachable set of the uncertain dynamic system. For example, in~\cite{alamoRobustMPCConstrained2005}, an approach was introduced in which the reachable set for additive uncertainties was over-approximated via interval arithmetics. In tube-based linear MPC, the reachable sets (or tubes) can be computed offline or online as polytopes~\cite{leeRobustRecedingHorizon2000, flemingRobustTubeMPC2015a}. For nonlinear systems, the computations become increasingly complex, but there exist some approaches based, for example, on ellipsoids~\cite{houskaRobustOptimizationNonlinear2012,villanuevaComputingEllipsoidalRobust2017,fengMinmaxDifferentialInequalities2019, messererEfficientAlgorithmTubebased2021}, see also~\cite{kouvaritakisModelPredictiveControl2016} or~\cite{mayneTubebasedRobustNonlinear2011} for an overview.

This paper is motivated by scenario-tree approaches for handling parametric or additive uncertainties over the prediction horizon~\cite{scokaertMinmaxFeedbackModel1998,delapenaStochasticProgrammingApplied2005,luciaMultistageNonlinearModel2013c}. While these methods incorporate recourse by considering different input vectors for each branch, their computational complexity grows exponentially with the prediction horizon and uncertainty set dimension. Simplified approaches, such as considering a few key scenarios or combining with tube-based methods~\cite{subramanianTubeenhancedMultistageMPC2021}, have shown promise in practice, even for complex nonlinear systems~\cite{luciaMultistageNonlinearModel2013c,holtorfMultistageNMPCOnline2019,bonzaniniFastApproximateLearningbased2021a}.

Motivated by this, we focus on simple scenario-trees for robust control of monotone systems. Monotone systems~\cite{angeliMonotoneControlSystems2003} allow efficient computation of reachable sets, independent of the number of uncertainties, due to their dynamics being influenced monotonically by states and uncertainties~\cite{meyerRobustControlledInvariance2016,luciaSetBasedOptimalControl2016}. This enables robust control by enveloping the system's dynamics with the maximum and minimum realizations of the uncertainty.

 Monotone systems are present in many engineering fields, such as biochemistry\cite{luciaSetBasedOptimalControl2016} or temperature control in buildings~\cite{meyerRobustControlledInvariance2016}.
The concept of monotonicity can be extended by decomposing the dynamics into monotonically increasing and decreasing components~\cite{cooganMixedMonotonicityReachability2020, smithDiscreteDynamicsMonotonically2006} for reachability analysis and robust control of general nonlinear dynamic systems~\cite{abateRobustlyForwardInvariant2022,harapanahalliContractionGuidedAdaptivePartitioning2023}.

This paper is an extension of our previous work~\cite{heinleinRobustMPCApproaches2022}. 
The primary contribution of this work is to extend the presented methods for discretizing the optimal feedback strategy via partitioning of reachable sets, previously limited to monotone systems, to a general framework for guaranteed robust nonlinear MPC using the concept of mixed-monotonicity. By considering systems under arbitrary feedback policies, this framework also encompasses existing tube-based approaches. Additionally, the proposed methodology is demonstrated on a nonlinear system with up to 25 states and 20 uncertainties.

This paper is structured as follows. Section~\ref{sec:Monotone_Systems} introduces the concept of monotone systems and the calculation of reachable sets. The open-loop and closed-loop robust approaches for monotone systems are shown in Section~\ref{sec:Mon_MPC}. %
Section~\ref{sec:Mixed_Mon} extends the applicability of the %
presented approaches beyond monotone systems via the concept of mixed-monotonicity.
In Section~\ref{sec:Case}, the applicability of the proposed approach to large-scale nonlinear systems is demonstrated.
\vspace{-1ex}
\section{Reachable Sets of Monotone Systems} \label{sec:Monotone_Systems}

We consider nonlinear discrete-time systems of the form
\begin{equation} \label{eq:system}
\vx_{k+1}=\vf(\vx_k,\vu_k,\vp_k),
\end{equation}
where $\vx_k\in\sR^{n_x}$ denotes the states, $\vu_k \in \sR^{n_u}$ denotes the inputs and $\vp_k\in\sR^{n_p}$ represents the parameters assumed to be in a compact set $\sP$. The system dynamics $\vf : \mathbb{R}^{n_x}\times\mathbb{R}^{n_u}\times\mathbb{R}^{n_p}\mapsto \mathbb{R}^{n_x}$ is assumed to be continuous and differentiable $\forall \vx_k\in\sR^{n_x},\forall \vu_k\in\sR^{n_u},\forall \vp_k\in\sR^{n_p}$. 
\begin{definition}[Monotonicity of dynamic systems] \label{def:monotone}
A system is called monotone on the sets $\sX \subseteq\sR^{n_x},\ \sU\subseteq\sR^{n_u},\ \sP\subset\sR^{n_p}$ with respect to the states, if for every pair $ \hat{\vx}$ and $\tilde{\vx}$ in $\sX$ that satisfies the condition $\hat{\vx} \geq \tilde{\vx}$, the following inequality holds:
\begin{equation} \label{eq:monotone_states}
\vf(\hat{\vx},\vu,\vp)\geq \vf(\tilde{\vx},\vu,\vp),\quad \forall \vu\in\sU,\ \forall \vp\in\sP,
\vspace{-0.5ex}
\end{equation}
where the inequalities are understood elementwise.

Analogously, a system is called monotone on $\sX\subseteq\sR^{n_x},\ \sU\subseteq\sR^{n_u},\ \sP\subset\sR^{n_p}$ with respect to the uncertainty, if for every pair $\hat{\vp}$ and $\tilde{\vp}$ in $\sP$  that satisfies the condition $\hat{\vp} \geq \tilde{\vp}$, the following inequality holds:
\begin{equation} \label{eq:monotone_uncert}
\vf(\vx,\vu,\hat{\vp})\geq \vf(\vx,\vu,\tilde{\vp}),\quad \forall \vu\in\sU,\forall \vx\in\sX.
\end{equation}
\end{definition}
\vspace{-1ex}
\begin{remark}
In this paper inequalities over multiple dimensions are understood to hold independently in each dimension.
In the following we assume monotonicity with respect to the states and the uncertainties when mentioning monotone systems.
\end{remark}
\vspace{-0.5ex}
The monotonicity conditions in \eqref{eq:monotone_states} and \eqref{eq:monotone_uncert} can be also checked if all the elements of the Jacobian of the dynamics with respect to the states and the uncertainties are non-negative.
How to extend the idea of monotonicity to general systems is discussed in Section~\ref{sec:Mixed_Mon}.
For further details on monotone systems, see~\cite{angeliMonotoneControlSystems2003}.

An important advantage of monotone systems is that computing tight outer approximations of reachable sets is straightforward, as stated in the following proposition.
\begin{proposition} \label{prop:Reach_Sets}
The 1-step reachable set for any fixed input $\vu\in\sU$ of system~\eqref{eq:system} satisfying~\eqref{eq:monotone_states} and~\eqref{eq:monotone_uncert}
with $\vx\in \left[\vx^-,\vx^+\right]$ and $\vp\in \left[\vp^-,\vp^+\right]$ is bounded by the interval 
\begin{equation} \label{eq:Reach_sets}
\vf(\vx,\vu,\vp) \in \left[\vf(\vx^-,\vu,\vp^-),\vf(\vx^+,\vu,\vp^+)\right].
\end{equation}
\end{proposition}
Proposition~\ref{prop:Reach_Sets} can be proven directly by applying Definition~\ref{def:monotone}.
The term hyperrectangle is used as a synonym for a multidimensional interval spanned by two points called the bottom left and top right corners of the hyperrectangle.
The set in~\eqref{eq:Reach_sets} is the tightest hyperrectangular outer approximation of the reachable set, as its top right and bottom left corners are the corners of the true reachable set.

\section{Robust MPC for monotone systems} \label{sec:Mon_MPC}
The computationally efficient tight outer approximation of reachable sets for a fixed input for monotone systems can be leveraged to formulate an MPC approach which is robust against uncertainties $\vp\in\sP = \left[ \vp^-, \vp^+\right]$, as done in our previous work~\cite{heinleinRobustMPCApproaches2022}.
At each time step in the prediction horizon $N$, the reachable set is approximated via a hyperrectangle according to Proposition~\ref{prop:Reach_Sets}. 
To obtain an open-loop robust MPC approach for monotone systems, a single input trajectory is applied for both extreme values of the uncertainty. The term \emph{open-loop} denotes that the method does not consider that one can react to future disturbances via future measurements in the closed-loop application~\cite{leeWorstcaseFormulationsModel1997}. %
The main advantage of the exploitation of monotonicity properties to formulate a robust MPC controller is that the number of variables and constraints 
scales only linearly in the state dimension $n_x$ as well as the prediction horizon $N$. 
However, an important drawback of the open-loop controller in~\cite{heinleinRobustMPCApproaches2022} is that it is open-loop robust, which can lead to conservative feasible domains and overly risk-averse performance~\cite{leeWorstcaseFormulationsModel1997}. 
Especially the assumption on the terminal set can be difficult to fulfill. It is required that there exists an input that counteracts the growth of the reachable set for both the maximum and the minimum realizations of the disturbance. This requirement may be impossible for unstable systems with additive disturbances, if recourse is not considered.
\begin{figure}
    \centering
    \includegraphics[width=0.44\textwidth]{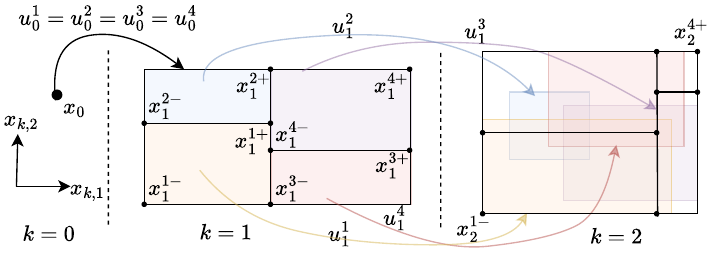}
    \caption{Schematic representation of approach~\eqref{eq:Closed_Loop} in a phase plot for multiple time steps $k$ in the prediction horizon of a 2D system with 4 subregions. Each subregion $\left[\vx_k^{s-}, \vx_k^{s+} \right]$ is propagated with the input $\vu_k^s$. The subsets in the next time step are aligned to bound all propagated sets.}
    \label{fig:scheme}
    \vspace{-2ex}
\end{figure}

To introduce recourse into the robust MPC for monotone systems and thus transforming it into a \emph{closed-loop} robust approach, we propose to divide the reachable sets that are computed via Proposition~\ref{prop:Reach_Sets} into multiple hyperrectangles, as it is shown in the left rectangle in Figure~\ref{fig:scheme}.
The exact position of the division is left as a degree of freedom to the optimizer.
A different control input is associated for each subregion, introducing feedback into the method.
In each time step, we start with one hyperrectangle, which is partitioned into multiple subregions. 
The subregions are propagated with individual inputs (arrows in Fig.~\ref{fig:scheme}) via Proposition~\ref{prop:Reach_Sets}. All propagations of the subregions are bounded by one hyperrectangle (right black rectangle in Fig.~\ref{fig:scheme}).
This bounding hyperrectangle is used as an over-approximation of the reachable set for each next time step and will again be partitioned.
The bounding prevents the exponential complexity growth typical of scenario-trees, but can lead to additional conservatism, as the corners of this bounding hyperrectangle may not be reachable as depicted in Figure~\ref{fig:scheme} ($\vx_2^{4+}$ is not reachable). 
The main advantage of the strategy is that it results in linear growth of the complexity with the prediction horizon.
The number of partitions per dimension generates a total number of subregions $\mu_s$. 
In the proposed closed-loop approach $\mathcal{P}_N^{\text{CL}}(\vx_0)$, %
presented in~\eqref{eq:Closed_Loop}, each subregion is denoted by the superscript $s\in\sS=\{1,...,\mu_s\}$.
The open-loop approach, denoted $\mathcal{P}_N^{\text{OL}}(\vx_0)$, can be recovered by setting $\mu_s=1$. The optimization problem solved at each sampling time to obtain the proposed closed-loop robust MPC for monotone systems is:
\begin{subequations} \label{eq:Closed_Loop}
\begin{align}
\noalign{$\displaystyle \min_{\vx_{\left[0:N\right]}^{s\pm},\vu_{\left[0:N-1\right]}^s,\forall s\in\sS}  \quad J(\vx_{\left[0:N\right]}^{\left[1:\mu_s\right]\pm},\vu_{\left[0:N-1\right]}^{\left[1:\mu_s\right]})$}\\
	\text{s.t}:\ &\vx_0^{s\pm}=\vx_0, \label{eq:Closed_Loop:subeq:IS}\\
    &\vx_{k+1}^{\mu_s+}\geq \vf(\vx_k^{s+},\vu_k^s,\vp^{+}),\label{eq:Closed_Loop:subeq:Bounding+}\\
	&\vx_{k+1}^{1-}\leq \vf(\vx_k^{s-},\vu_k^s,\vp^{-}),\label{eq:Closed_Loop:subeq:Bounding-}\\
	&\left[\vx_k^{1-}, \vx_k^{\mu_s+}\right]\subseteq \sX,\ \label{eq:Closed_loop:subeq:State_Constraint} \\%
	&\vu_k^{s}\in\sU,\ \\ %
	&\left[\vx_N^{1-}, \vx_N^{\mu_s+}\right]\subseteq \sX_f, \label{eq:Closed_Loop:subeq:terminal_Constr}\\
	&\vu_0^1=\vu_0^s, \label{eq:Closed_Loop:subeq:In. Input Constr}\\
    &\vh(\vx_k^{\left[1:\mu_s\right]\pm})\leq\vzero, \label{eq:Closed_Loop:subeq:dividing_RS}
\end{align}
\end{subequations}
where the constraints involving $k$ or $s$ hold $\forall k\in\{0,...,N-1\},\ \forall s \in\sS$.
For each subregion $s$ and each step $k$, two points in the state space $\vx_k^{s+},\vx_k^{s-}$, are defined to span the corresponding hyperrectangles. These points, along with the associated input $\vu_k^s$, are included in the objective $J$.
\begin{remark}
For notational compactness, constraints that hold for both the top right and bottom left corners, e.g.~\eqref{eq:Closed_Loop:subeq:IS}, are concatenated in a slight abuse of notation via $\pm$.
\end{remark}
Via~\eqref{eq:Closed_Loop:subeq:Bounding+} and~\eqref{eq:Closed_Loop:subeq:Bounding-}, the bounding hyperrectangle of the next time step is constrained to contain all reachable sets originating from the propagation of the current subregions.
The hyperrectangle spanned in the last step of the prediction horizon is constrained in~\eqref{eq:Closed_Loop:subeq:terminal_Constr} to lie within the terminal set.
In addition,~\eqref{eq:Closed_Loop:subeq:In. Input Constr} equates all initial inputs, as the initial state is fixed.
The function $\vh(\vx_k^{\left[1:\mu_s\right]\pm})$ in~\eqref{eq:Closed_Loop:subeq:dividing_RS} denotes the constraints for each partition, that is, they enforce that the subregions at each time-step in the predictions are aligned next to each other without overlap and fill the whole bounding hyperrectangle obtained from the propagation of the previous step.
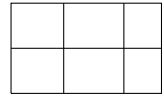
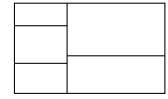
\begin{figure}
    \centering
    \begin{subfigure}{.24\textwidth}
        \centering
        \begin{tikzpicture}
            \draw (0,0) -- (2,0) -- (2,1.2) -- (0,1.2) -- (0,0);
            \draw (0.7,0) -- (0.7,1.2);
            \draw (1.5,0) -- (1.5,1.2);
            \draw (2,0.6) -- (0,0.6);
        \end{tikzpicture}
        \caption{Grid-like partitions}
        \label{fig:Cutting:subfig:Grid}
    \end{subfigure}%
    \begin{subfigure}{.24\textwidth} 
        \centering
        \begin{tikzpicture}
            \draw (0,0) -- (2,0) -- (2,1.2) -- (0,1.2) -- (0,0);
            \draw (0.7,0) -- (0.7,1.2);
            \draw (0, 0.4) -- (0.7, 0.4);
            \draw (0, 0.9) -- (0.7, 0.9);
            \draw (0.7, 0.5) -- (2, 0.5);
        \end{tikzpicture}
        \caption{Search-tree-like partitions}
        \label{fig:Cutting:subfig:Tree}
    \end{subfigure}
    \caption{Methods on dividing the reachable set} \label{fig:Cutting}
    \vspace{-2ex}
\end{figure}
The degrees of freedom in the arrangement of the subregions is slightly increased compared to~\cite{heinleinRobustMPCApproaches2022}. In~\cite{heinleinRobustMPCApproaches2022}, the grid-like partitioning sketched in Figure~\ref{fig:Cutting:subfig:Grid} was presented. In this paper, a search-tree like strategy, sketched in Figure~\ref{fig:Cutting:subfig:Tree}, is employed, which allows the independent further division of already separated subregions.
Regardless, the arrangement of the subregions is constrained in $\vh$ by constraints of the form $\vx_{k,d}^{q+}-\vx_{k,d}^{t+}$, $\vx_{k,d}^{q-}-\vx_{k,d}^{t+}$ and $\vx_{k,d}^{q-}-\vx_{k,d}^{t-}$, where $q$ and $t$ belong to subsets of $\sS$. For Figure~\ref{fig:scheme}, the constraints are
\begin{align*}
&x_{k,2}^{1-} = x_{k,2}^{3-} &\quad& x_{k,2}^{4+} = x_{k,2}^{2+} &\quad& x_{k,1}^{1-} = x_{k,1}^{2-} &\quad& x_{k,2}^{1+} = x_{k,2}^{2-}, \\
&x_{k,1}^{1+} = x_{k,1}^{2+} &\quad& x_{k,1}^{1+} = x_{k,1}^{3-} &\quad& x_{k,1}^{3-} = x_{k,1}^{4-} &\quad& x_{k,2}^{3+} = x_{k,2}^{4-}, \\
&x_{k,1}^{3+} = x_{k,1}^{4+} &\quad& x_k^{s\pm} \geq x_k^{1-}  &\quad& x_k^{4+} \geq x_k^{s\pm}, &\, &\forall s \in \sS.
\end{align*}
A thorough mathematical and algorithmical explanation is given in the supplementary repository$^\text{{\ref{URL}}}$.

The following assumptions are taken to prove recursive feasibility and robust constraint satisfaction.
\begin{assumption} \label{ass:Monotonicity}
The controlled system~\eqref{eq:system} is monotone in the states and the uncertainties, so~\eqref{eq:monotone_states} and~\eqref{eq:monotone_uncert} hold.%
\end{assumption}
\begin{assumption} \label{ass:Box_constraints}
The state constraints $\sX$ are given as box constraints, that is, as a hyperrectangle:
\begin{equation}\label{eq:Interval}
    \sX=\left[\vx^{\textnormal{min}},\vx^{\textnormal{max}}\right]=\{ \vx\in \sR^{n_x} | \vx^{\textnormal{min}}\leq \vx \leq \vx^{\textnormal{max}} \},
\end{equation}
and the uncertainty $\vp$ takes values in the hyperrectangle $\vp\in\sP=\left[\vp^-,\vp^+\right]$.
\end{assumption}
The assumptions on the terminal set, displayed in Figure~\ref{fig:RCIS}, can be relaxed compared to the open-loop approach, as recourse can be introduced in the terminal set of the closed-loop approach via the partitioning of the state space.
\begin{assumption} \label{ass:RCIS_closed_loop}
    The terminal set $\sX_f\subseteq \sX$ is a hyperrectangular robust control invariant set, spanned by $\left[\vx_{\textnormal{RCIS}}^{1-}, \vx_{\textnormal{RCIS}}^{\mu_s+}\right]$, divided into subregions $\mu_s$ that satisfy $\vh(\vx_{\textnormal{RCIS}}^{\left[1:\mu_s\right]\pm})\leq\vzero$ so that the following holds for $\left[\vx_{\textnormal{RCIS}}^{1-}, \vx_{\textnormal{RCIS}}^{\mu_s+}\right]$
    \begin{equation*}
        \exists \vu_{\textnormal{RCIS}}^{\left[1:\mu_s\right]}\in\sU : \vf(\vx_{\textnormal{RCIS}}^{s\pm},\vu_{\textnormal{RCIS}}^s,\vp^{\pm})\in \left[\vx_{\textnormal{RCIS}}^{1-},\vx_{\textnormal{RCIS}}^{\mu_s+}\right] \forall s\in\sS.
    \end{equation*}
\end{assumption}
Here, we present the assumption of a hyperrectangular terminal set for simplicity. The assumption on the shape of $\sX_f$ can be relaxed
by replacing~\eqref{eq:Closed_Loop:subeq:terminal_Constr} with the constraints
\begin{align} \label{eq:RCIS_relax}
    \left[\vx_N^{s-}, \vx_N^{s+}\right]\subseteq \left[\vx_{N-1}^{1-}, \vx_{N-1}^{\mu_s+}\right], \forall s\in \sS.
\end{align}
See Section~VI of \cite{heinleinRobustMPCApproaches2022} for details on how to compute a set that satisfies Assumption~\ref{ass:RCIS_closed_loop}.\\
The solution $\vx_{\left[0:N\right]}^{s\pm*},\vu_{\left[0:N-1\right]}^{s*}$ of the optimization problem $\mathcal{P}_N^{\text{CL}}(\vx)$ results in the control law $\pi_N^{\text{CL}}(\vx)=\vu_0^{0*}$, which leads to the following properties of the closed-loop system.
\begin{theorem} \label{thm:Closed_loop}
Let Assumptions~\ref{ass:Monotonicity},~\ref{ass:Box_constraints} and~\ref{ass:RCIS_closed_loop} hold.
The closed-loop system obtained by applying the MPC control law $\pi_N^{\textnormal{CL}}(\vx)$, obtained from~\eqref{eq:Closed_Loop}, to the system~\eqref{eq:system} satisfies the input constraints $\vu_k\in\sU$ and state constraints $\vx_k\in\sX$ for any $\vp_k\in\sP,\ \forall k\geq0$.
The MPC problem $\mathcal{P}_N^{\textnormal{CL}}(\vx_k)$ is recursively feasible.%
\end{theorem}
For the proof see Theorem 2 in~\cite{heinleinRobustMPCApproaches2022}.
\begin{remark} \label{rem:Box_constraints}
    The requirement of box state-constraints given by Assumption~\ref{ass:Box_constraints} can be relaxed. Constraint~\eqref{eq:Closed_loop:subeq:State_Constraint} can be checked by vertex enumeration for any convex $\sX$, which is not further discussed as it scales exponentially in $n_x$. For polytopic constraints, Farkas Lemma~\cite{kouvaritakisModelPredictiveControl2016} can ensure that the reachable sets lie within the state constraints with the same amount of constraints. This is omitted here for simplicity.
    \vspace{-1ex}
\end{remark}
The complexity of the approach grows linearly with the number of subregions, due to the bounding hyperrectangle. The conservatism introduced by the bounding hyperrectangle is a drawback most tube-based approaches suffer from, as the shape of the reachable set is over-approximated by a simpler shape, which is easier to propagate.
The number of subregions can be chosen dependent on the method of dividing the reachable sets.
If the reachable set is divided fully in each dimension, the number of subregions grows exponentially. However, it is often sufficient to divide the reachable sets only in some dimensions to achieve significant performance improvement, when compared to the open-loop approach $\mathcal{P}_N^{\text{OL}}(\vx)$, as is shown in the example in Section~\ref{sec:Case}.
The number of subregions is a tuning parameter of the approach. In our experience, and as shown in the case study in Section~\ref{sec:Case}, it is often sufficient to have a low number of partitions to achieve a significant performance improvement compared to an open-loop approach.
\begin{figure}
    \centering
        \includegraphics[width=0.40\textwidth]{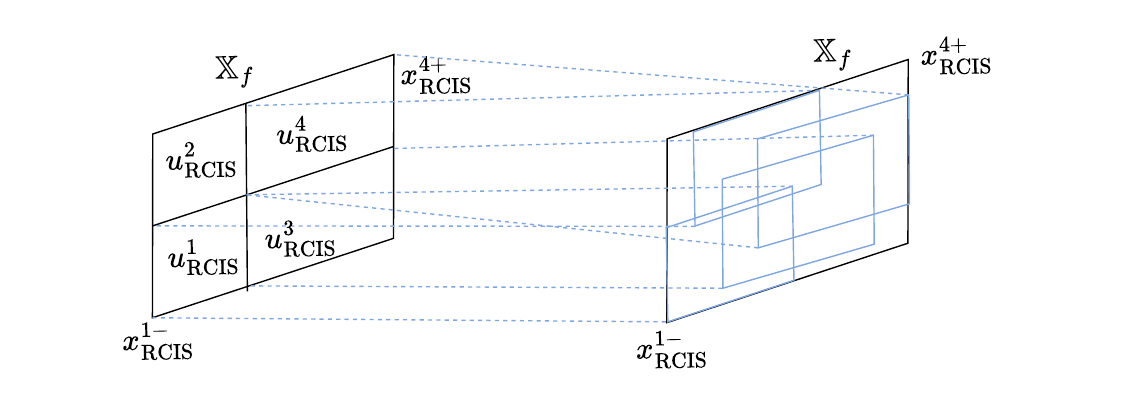}
    \caption{Visualisation of the assumption on the terminal set $\sX_f$. The rectangle $\sX_{f}=\left[x_{\text{RCIS}}^{1-},\vx_{\text{RCIS}}^{4+}\right]$ is divided and each subregion is propagated with an individual input. The propagations need to lie inside $\sX_f$.} \label{fig:RCIS}
    \vspace{-3ex}
\end{figure}
\vspace{-1ex}
\section{Extension to general systems and feedback policies} \label{sec:Mixed_Mon}
\vspace{-1ex}
The results of the previous sections
hold only when Assumption~\ref{ass:Monotonicity} holds, and the system is monotone. Through a state transformation some non-monotone systems can be made monotone~\cite{kalliesMonotonicityKineticProofreading2016, abateImprovingFidelityMixedMonotone2021}.
In addition, we extend in the following our approach to general nonlinear systems and systems under feedback policies.
\vspace{-1ex}
\subsection{Generalization via mixed-monotonicity}
To extend the approach to non-monotone systems, we exploit the concept of mixed-monotonicity, which describes that a general nonlinear system can be decomposed into monotonously increasing and decreasing parts by defining a suitable decomposition function~\cite{smithDiscreteDynamicsMonotonically2006}. This can be formally stated by the following definition from the work in~\cite{cooganMixedMonotonicityReachability2020, abateVerificationRuntimeAssurance2021}.
\vspace{-1ex}
\begin{definition}[Decomposition function of a dynamic system] \label{def:decomp}
Given a continuous function $\vd:\sX\times\sP\times\sU\times\sX\times\sP \mapsto \sR^{n_x}$ the system~\eqref{eq:system} is mixed-monotone with respect to $\vd$, if $\forall u\in\sU$
\begin{enumerate}
    \item $\vf(\vx,\vu,\vp)=\vd(\vx,\vp,\vu,\vx,\vp),\ \forall \vx\in\sX,\ \forall \vp\in\sP$,
    \item $\vd(\tilde{\vx},\tilde{\vp},\vu,\overline{\vx},\overline{\vp}) \leq \vd(\hat{\vx},\hat{\vp},\vu,\overline{\vx},\overline{\vp}),$ $\forall \tilde{\vx}, \overline{\vx}, \hat{\vx}$ $ \in\sX$, such that $\tilde{\vx}\leq\hat{\vx}$ and $\ \forall \tilde{\vp}, \overline{\vp}, \hat{\vp} \in\sP$, such that $\tilde{\vp}\leq \hat{\vp}$,
    \item $\vd(\tilde{\vx},\tilde{\vp},\vu,\breve{\vx},\breve{\vp}) \leq \vd(\tilde{\vx},\tilde{\vp},\vu,\overline{\vx},\overline{\vp}),$ $\forall \tilde{\vx},\overline{\vx},\breve{\vx}$ $\in\sX$, such that $\overline{\vx}\leq\breve{\vx}$ and $\ \forall \tilde{\vp},\overline{\vp},\breve{\vp}\in\sP$, such that $\overline{\vp}\leq \breve{\vp}$.
\end{enumerate}
\end{definition}
The concept of decomposition functions enables the simple computation of approximate reachable sets of general systems~\cite{cooganMixedMonotonicityReachability2020}:
\begin{proposition} \label{prop:Reach_Sets_Mixed_Mon}
The 1-step reachable set for any fixed input $\vu\in\sU$ of a discrete system~\eqref{eq:system} with a decomposition function $\vd$ according to Def.~\ref{def:decomp} with $\vx\in \left[\vx^-,\vx^+\right]$ and $\vp\in \left[\vp^-,\vp^+\right]$ is bounded by the multidimensional interval 
\begin{equation*} \label{eq:Reach_sets_Mixed_Monotone}
\vf(\vx,\vu,\vp) \in \left[\vd(\vx^-,\vp^-,\vu,\vx^+,\vp^+),\vd(\vx^+,\vp^+,\vu,\vx^-,\vp^-)\right].
\end{equation*}
\end{proposition}
Similar to Proposition~\ref{prop:Reach_Sets}, the proof of Proposition~\ref{prop:Reach_Sets_Mixed_Mon} is a direct result of Definition~\ref{def:decomp}, as $d(\vx,\vp,\vu,\cdot,\cdot)$ and $d(\cdot,\cdot,\vu,\vx,\vp)$ over-approximate the monotonic increasing and decreasing parts of~\eqref{eq:system}~\cite{cooganMixedMonotonicityReachability2020}.
Our proposed method in this work uses the idea of mixed-monotonicity to formulate an open-loop or closed-loop robust MPC approach as the one in~\eqref{eq:Closed_Loop} that scales only linearly with the prediction horizon and the number of states and does not scale with the number of uncertainties.
To formulate the approach for general systems, the following is assumed.
\begin{assumption} \label{ass:decomp}
The function $\vd$ is a decomposition function according to Definition~\ref{def:decomp} of the system~\eqref{eq:system}. Additionally, there exists a hyperrectangular terminal set satisfying Assumption~\ref{ass:RCIS_closed_loop} with respect to the decomposition function instead of the system function, that is:
\begin{equation*}
        \exists \vu_{\textnormal{RCIS}}^{\left[1:\mu_s\right]}\in\sU : \vd(\vx_{\textnormal{RCIS}}^{s\pm},\vp^{\pm},\vu_{\textnormal{RCIS}}^s,\vx_{\textnormal{RCIS}}^{s \mp},\vp^{\mp})\in \sX_f,\  \forall s\in\sS.
    \end{equation*}
\end{assumption}
We introduce the problem $\mathcal{P}_N^{\text{CL}_{\vd}}(\vx_0)$ for general systems with a decomposition function $\vd$
\begin{subequations} \label{eq:Closed_Loop_Tube}
\begin{align}
\noalign{$\displaystyle \min_{\vx_{\left[0:N\right]}^{s\pm},\vu_{\left[0:N-1\right]}^s,\forall s\in\sS}  \quad J(\vx_{\left[0:N\right]}^{\left[1:\mu_s\right]\pm},\vu_{\left[0:N-1\right]}^{\left[1:\mu_s\right]})$}\ \\
	\text{s.t}:\ &\vx_0^{s\pm}=\vx_0, \label{eq:Closed_Loop_Tube:IS}\\
    &\vx_{k+1}^{\mu_s+}\geq \vd(\vx_k^{s+},\vp^{+},\vu_k^s,\vx_k^{s-},\vp^{-}),\label{eq:Closed_Loop_Tube:subeq:Bounding+}\\
	&\vx_{k+1}^{1-}\leq \vd(\vx_k^{s-},\vp^{-},\vu_k^s,\vx_k^{s+},\vp^{+}),\label{eq:Closed_Loop_Tube:subeq:Bounding-}\\
	&\left[\vx_k^{1-}, \vx_k^{\mu_s+}\right]\subseteq \sX,\ \label{eq:Closed_Loop_Tube:subeq:State_Constraint} \\%
	&\vu_k^{s}\in\sU,\ \label{eq:Closed_Loop_Tube:Input_Constraint} \\
	&\left[\vx_N^{1-}, \vx_N^{\mu_s+}\right]\subseteq \sX_f, \label{eq:Closed_Loop_Tube:subeq:terminal_Constr}\\
	&\vu_0^1=\vu_0^s, \label{eq:Closed_Loop_Tube:subeq:In. Input Constr}\\
    &\vh(\vx_k^{\left[1:\mu_s\right]\pm})\leq\vzero, \label{eq:Closed_Loop_Tube:subeq:dividing_RS}
\end{align}
\end{subequations}
where the constraints involving $k$ or $s$ hold $\forall k\in\{0,...,N-1\},\ \forall s \in\sS$.
Naturally, it is also possible to use no partitioning, i.e. $\mu_s=1$, in~\eqref{eq:Closed_Loop_Tube}, which will be denoted as $\mathcal{P}_N^{\text{OL}_{\vd}}(\vx_0)$. 
The solution $\vx_{\left[0:N\right]}^{s\pm*},\vu_{\left[0:N-1\right]}^{s*}$ of the optimization problem $\mathcal{P}_N^{\text{CL}_{\vd}}(\vx)$ results in the control law $\pi_N^{\text{CL}_{\vd}}(\vx)=\vu_0^{0*}$, which leads to the following properties of the closed-loop system.
\begin{theorem} \label{thm:Closed_loop_Mixed}
Let Assumptions~\ref{ass:Box_constraints} and~\ref{ass:decomp} hold for system~\eqref{eq:system}.
The closed-loop system obtained by applying the MPC control law $\vu_k=\pi_N^{\textnormal{CL}_{\vd}}(\vx_k)$, obtained from~\eqref{eq:Closed_Loop_Tube}, to the system~\eqref{eq:system} satisfies the input constraints $\vu_k\in\sU$ and state constraints $\vx_k\in\sX$ for any $\vp_k\in\sP,\ \forall k\geq0$.
The MPC problem $\mathcal{P}_N^{\textnormal{CL}_{\vd}}(\vx_k)$ is recursively feasible.
\end{theorem}
\begin{proof}
    According to Assumption~\ref{ass:decomp} and Proposition~\ref{prop:Reach_Sets_Mixed_Mon}, the reachable sets of the system for each time step $k$ and subregion $s$ $\left[\vx_k^{s-}, \vx_k^{s+}\right]$ are bounded by $\left[\vx_{k+1}^{1-},\vx_{k+1}^{\mu_s+}\right]$ via~\eqref{eq:Closed_Loop_Tube:subeq:Bounding+} and~\eqref{eq:Closed_Loop_Tube:subeq:Bounding-}. Since Assumption~\ref{ass:Box_constraints} holds, the state constraints can be checked via the corners of the hyperrectangle $\left[\vx_k^{1-}, \vx_k^{\mu_s+}\right]$ as in~\eqref{eq:Closed_Loop_Tube:subeq:State_Constraint}. The input constraints are satisfied via~\eqref{eq:Closed_Loop_Tube:Input_Constraint}.  Therefore, state and input constraint satisfaction is guaranteed.
    Recursive feasibility can be shown analogously to the proof in~\cite{heinleinRobustMPCApproaches2022}, as at the next time steps the propagated reachable sets will lie within the shifted reachable sets from the previous iteration. The division of the reachable sets into subregions can always align with the previous divisions, such that the previous computed inputs can be reused. At the end of the prediction horizon, Assumption~\ref{ass:decomp} guarantees robust invariance, such that for any hyperrectangular reachable set inside the terminal set, there exists a division into subregions with respective inputs, such that the bounding interval over all propagations lies in the terminal set as well. This specific partition and set of inputs are used at the end of the prediction horizon of the candidate feasible solution. Thus recursive feasibility is guaranteed.
\end{proof}
The goal of this work is to study recursive feasibility and robust constraint satisfaction. The input to state stability of the open-loop approach without the state partitioning can be inferred from~\cite{limonInputtoStateStabilityUnifying2009}. The extension to establish the stability properties of the presented closed-loop approach can be based on the ideas from~\cite{luciaStabilityPropertiesMultistage2020,subramanianTubeenhancedMultistageMPC2021} and is left as future work.

The price to pay for this generalization to any nonlinear system is that the reachable sets obtained are prone to the wrapping effect and obtaining a differentiable decomposition function can be difficult and case specific for general nonlinear systems. The wrapping effect occurs because the top right and bottom left corners of the hyperrectangular over-approximation of the reachable sets, which are further propagated, may lie outside the true reachable set.
Methods using more complex shapes of the reachable set, like ellipsoids~\cite{houskaRobustOptimizationNonlinear2012} or zonotopes~\cite{bravoRobustMPCConstrained2006} can partly avoid the wrapping effect at the cost of more complex schemes and more difficult to design boundaries for the reachable sets.

Regarding the feasibility of Assumption~\ref{ass:decomp}, it is possible to define the tightest possible decomposition function of a general dynamic system, as presented in~\cite{yangTightDecompositionFunctions2019}.
\begin{proposition} \label{prop:tight_decomp}
Any system~\eqref{eq:system} is mixed-monotone with respect to a decomposition function $\vd$ satisfying $\forall i=1,\dots,n_x$
\begin{multline} \label{eq:tight_decomp}
    \evd_i(\tilde{\vx},\tilde{\vp},\vu,\hat{\vx},\hat{\vp})= \\
    \begin{cases}
    \min\limits_{
    \substack{{\overline{\vx}\in\left[\tilde{\vx},\hat{\vx}\right], \overline{\vp}\in\left[\tilde{\vp},\hat{\vp} \right]}}} \evf_i(\overline{\vx},\vu,\overline{\vp}),\ \textnormal{if}\ \begin{bmatrix}
    \tilde{\vx} \\
    \tilde{\vp}
    \end{bmatrix} \leq \begin{bmatrix}
    \hat{\vx}\\
    \hat{\vp}
    \end{bmatrix} \\%
    \max\limits_{
    \substack{{\overline{\vx}\in\left[\hat{\vx},\tilde{\vx}\right], \overline{\vp}\in\left[\hat{\vp},\tilde{\vp} \right]}}} \evf_i(\overline{\vx},\vu,\overline{\vp}),\ \textnormal{if}\ \begin{bmatrix}
    \hat{\vx} \\
    \hat{\vp}
    \end{bmatrix} \leq \begin{bmatrix}
    \tilde{\vx}\\
    \tilde{\vp}
    \end{bmatrix} \\%
    \end{cases}
\end{multline}
For the continuous counterpart see~\cite{cooganMixedMonotonicityReachability2020}.
\end{proposition}
The decomposition function of monotone systems is the system function $f(\tilde{\vx},\vu,\tilde{\vp})$ itself.

To circumvent non-differentiabilities in the decomposition function for general nonlinear systems, which would lead to difficulties in a gradient based optimization scheme,~\eqref{eq:tight_decomp} can be relaxed by taking the minimum and maximum of individual terms similar to interval arithmetics~\cite{alamoRobustMPCConstrained2005}. 
\begin{remark} \label{rem:linear_decomp}
    For any linear dynamical system of the form $\vx_{k+1}=\mA\vx_k+\mB \vu + \mE\vp_k$, the decomposition function can be obtained by partitioning the matrices into positive and negative parts
    \begin{equation*}
    \evd(\tilde{\vx},\tilde{\vp},\vu,\hat{\vx},\hat{\vp}) = \mA^+ \tilde{\vx} + \mA^- \hat{\vx} + \mE^+ \tilde{\vp} + \mE^- \hat{\vp} + \mB \vu,
\end{equation*}
    where $\mA^+ +\mA^- =\mA$ and $\mA^+ \geq \vzero \geq \mA^-$ elementwise (analogously for $\mE$)~\cite{cooganMixedMonotonicityReachability2020}. 
\end{remark}
For nonlinear systems,~\eqref{eq:tight_decomp} can be relaxed similarly. For example, if the $i$-th element of~\eqref{eq:system} is the sum of the terms $\evf_i^a$ and $\evf_i^b$, the following decomposition function $\breve{\evd}_i$ still satisfies the requirements of Definition~\ref{def:decomp} when \(\tilde{\vx} \leq \hat{\vx}\) and \(\tilde{\vp} \leq \hat{\vp}\):
\begin{multline*}
    \breve{\evd}_i(\tilde{\vx},\tilde{\vp},\vu,\hat{\vx},\hat{\vp}) = \\
    \min\limits_{
    \substack{{\overline{\vx}\in\left[\tilde{\vx},\hat{\vx}\right], \overline{\vp}\in\left[\tilde{\vp},\hat{\vp} \right]}}} \evf_i^a(\overline{\vx},\vu,\overline{\vp}) + \min\limits_{
    \substack{{\overline{\vx}\in\left[\tilde{\vx},\hat{\vx}\right], \overline{\vp}\in\left[\tilde{\vp},\hat{\vp} \right]}}} \evf_i^b(\overline{\vx},\vu,\overline{\vp}) \\
    \leq  \min\limits_{\substack{{\overline{\vx}\in\left[\tilde{\vx},\hat{\vx}\right], \overline{\vp}\in\left[\tilde{\vp},\hat{\vp} \right]}}} \evf_i^a(\overline{\vx},\vu,\overline{\vp})+\evf_i^b(\overline{\vx},\vu,\overline{\vp})\\
    =\evd_i(\tilde{\vx},\tilde{\vp},\vu,\hat{\vx},\hat{\vp}),
\end{multline*}%
where the decomposition for the max operator is analogous when \(\hat{\vx} \leq \tilde{\vx}\) and \(\hat{\vp} \leq \tilde{\vp}\).
If $\evf_i^a$ and $\evf_i^b$ do not change signs, this can also be applied for the multiplication and division of $\evf_i^a$ and $\evf_i^b$. Decomposing the system function into individual terms, until the min and max operations become trivial leads to more conservative, but differentiable decomposition functions.
This is systematically applied in the nonlinear case-study presented in Section~\ref{sec:Case}, demonstrating that Assumption~\ref{ass:Monotonicity} is not needed anymore, widening the range of applicability of the proposed MPC approaches beyond monotone systems without additional computational complexity.

\subsection{Relation to other robust MPC approaches - general feedback policies}

Interestingly, this framework can be equivalent to well known tube-based approaches. 
Low complexity tube-based MPC, as presented in~\cite{kouvaritakisModelPredictiveControl2016}, propagates the tubes via the decomposition function for the linear system under affine feedback as in Remark~\ref{rem:linear_decomp}. 
This motivates the following consideration of general feedback policies in the mixed-monotone robust MPC approach for further possibilities to include recourse.

 The system function $\vf$ can already represent a system under a feedback policy $\vu=\kappa(\vx,\vv)$, so $\vf(\vx,\vu,\vp)=\vf_{\kappa}(\vx,\vv,\vp)$, where $\vv$ are parameters influencing the policy in every timestep, for example a feedforward term.
 The feedback policy can either be chosen to stabilize the system or to improve the tightness of a decomposition function by canceling non-monotone terms.

For simplicity, we denote in the following the state and input constraints $\vx \in \sX$, $\vu \in \sU$ as $\vg(\vx,\vu) \leq \vzero$. 
To account for the coupling of the state and input constraints via $\kappa(\vx,\vv)$, we need following assumption on the function $\vg(\vx,\kappa(\vx,\vv))$.
\begin{assumption} \label{ass:constraint_function}
    The function $\vd_{\vg}(\vx^1,\vv,\vx^2)$ is a decomposition function for  $\vg(\vx,\kappa(\vx,\vv)) $ with respect to $\vx$ and the feedback law $\kappa(\vx,\vv)$  according to Definition~\ref{def:decomp} without $\vp$.
\end{assumption}
To ensure constraint satisfaction for all  values that $\vu$ and $\vx$ can take within the reachable set for a fixed $\vv$, $\vd_{\vg}$ decomposes $\vg$ into increasing and decreasing parts with respect to $\vx$. Similar to Assumption~\ref{ass:decomp}, the existence of $\vd_{\vg}$ is guaranteed according to Proposition~\ref{prop:tight_decomp}. For polytopic constraints with affine feedback, $\vd_{\vg}$ can be obtained as in Remark~\ref{rem:linear_decomp}.\\
\begin{corollary}
Let Assumptions~\ref{ass:decomp} and~\ref{ass:constraint_function} hold for system~\eqref{eq:system} with respect to the feedback law $\kappa(\vx,\vv) $. 
The closed-loop system obtained by applying the MPC control law $\vu_k=\kappa(\vx_k,\vv_0^{0*})$, obtained from~\eqref{eq:Closed_Loop_Tube}  when using the decomposition function with respect to the system under feedback in~\eqref{eq:Closed_Loop_Tube:subeq:Bounding+} and~\eqref{eq:Closed_Loop_Tube:subeq:Bounding-} as well as replacing~\eqref{eq:Closed_Loop_Tube:subeq:State_Constraint} and~\eqref{eq:Closed_Loop_Tube:Input_Constraint} with $\vd_{\vg}(\vx_k^{\mu_s +},\vv_k^s,\vx_k^{1-})\leq\vzero$, to the system~\eqref{eq:system} satisfies the constraint $\vg(\vx_k,\vu_k)\leq0$ for any $\vp_k\in\sP,\ \forall k\geq0$.
The MPC problem is recursively feasible.
\end{corollary}
\begin{proof}
    The proof is analogous to Theorem~\ref{thm:Closed_loop_Mixed}, as $\vd_g$ guarantees the satisfaction of the coupled state and input constraints.
\end{proof}
This formulation further generalizes the presented approach. 
It is thus possible to guarantee robust constraint satisfaction and recursive feasibility for general systems that include a feedback law, as done in tube-based methods, or multiple feed-forward terms as done in~\eqref{eq:Closed_Loop}, or a combination of both. 
Interval arithmetic and differential inclusion methods can provide decomposition functions. Thus, the proposed formulation with general feedback policies extends approaches like~\cite{alamoRobustMPCConstrained2005}, introducing additional flexibility through partitioning strategies.
In comparison to linearization-based ellipsoidal
methods~\cite{houskaRobustOptimizationNonlinear2012, messererEfficientAlgorithmTubebased2021}, whose guarantees rely on ellipsoidal nonlinearity bounds, the proposed method exhibits linear scalability with the number of states. Moreover, unlike~\cite{kohlerComputationallyEfficientRobust2021}, it avoids the computation of Lyapunov-function-based nonlinearity bounds.
The primary advantage lies in its computational efficiency, as the complexity remains independent of the number of uncertainties. In contrast, scenario-based approaches~\cite{luciaMultistageNonlinearModel2013c}, while requiring minimal setup effort~\cite{fiedlerDompcFAIRNonlinear2023}, exhibit exponential growth in complexity with increasing uncertainties, rendering them unsuitable for the following case study.

\section{Nonlinear, non-monotone CSTR Cascade} \label{sec:Case}
To examine the scalability of the mixed-monotone robust MPC framework with the state dimension, a nonlinear, non-monotone model of a cascade of $n_R$ continuous stirred tank reactors (CSTR) is examined.
\begin{figure*}
    \centering
    \includegraphics[width=0.98\textwidth]{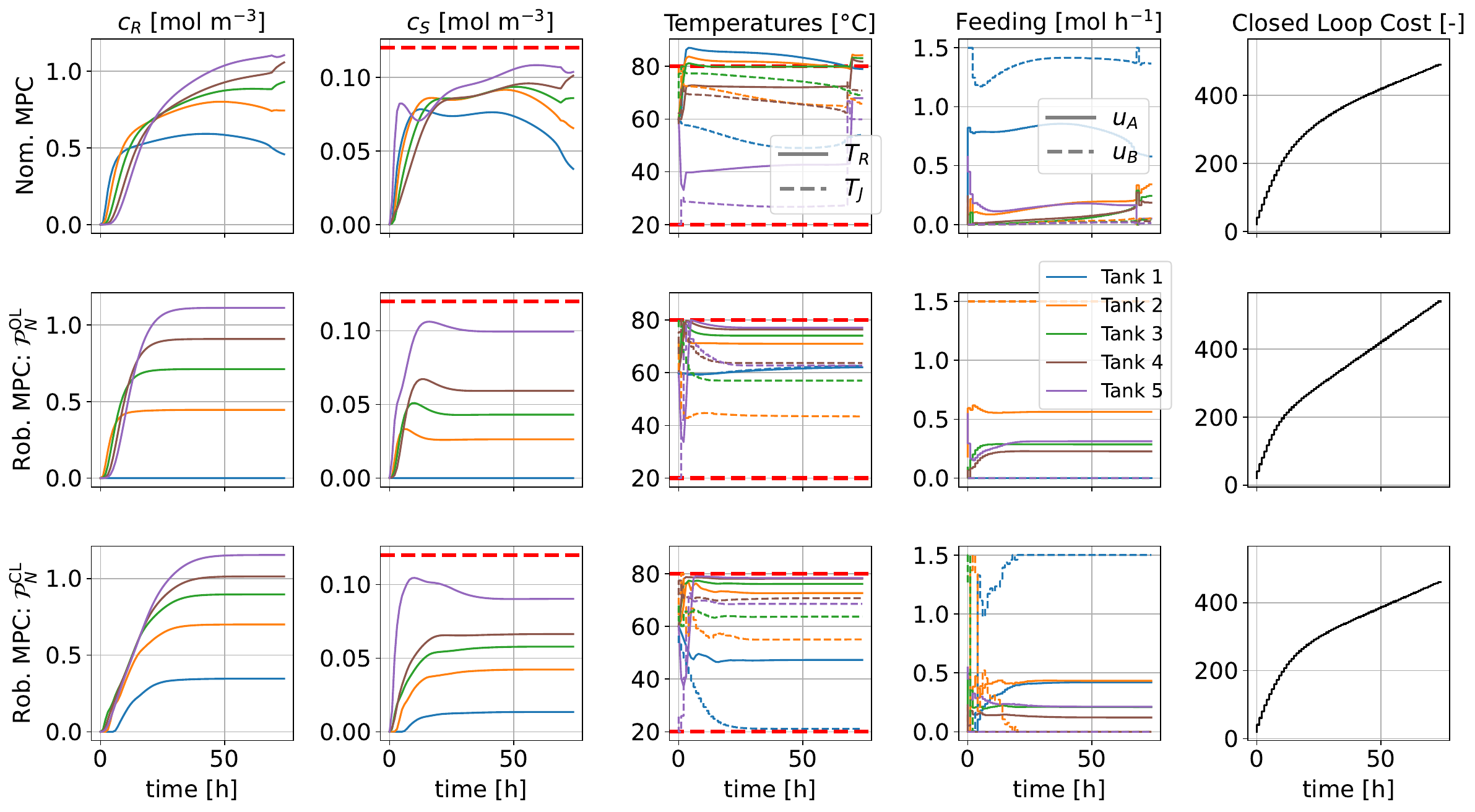}
    \caption{Closed-loop simulation of the five tank CSTR cascade. Nominal MPC (first row), the open-loop approach (second row) and the robust closed-loop approach (third row) are compared. The first two columns compare the concentrations of the value product $c_R$ and the side product $c_S$. The third column shows the reactor temperatures $T_{R,i}$ as solid lines and the jacket temperatures $T_{J,i}$ as dashed lines. The fourth column describes the inputs $u_{A,i}$ (solid) and $u_{B,i}$ (dashed) in each tank. The last column compares the accumulated closed-loop costs. The red dashed lines in all plots display the state and input constraints.}
    \label{fig:CSTR}
    \vspace{-2ex}
\end{figure*}
Educt $A$ reacts with $B$ to form $R$ with the second order side reaction of $2A\rightarrow S$. The differential equations of the $i$th CSTR are with the reaction rates $r_{1,i}=k_{1,i}\exp{-\frac{E_{A1,i}}{R_\text{gas}(T_{R,i}+273.15)}}c_{A,i}c_{B,i}$ and $r_{2,i}=k_{2,i}\exp{-\frac{E_{A2,i}}{R_\text{gas}(T_{R,i}+273.15)}}c_{A,i}^2$
%
%
%
%
%
%
%
%
%
%
%
%
%
%
%
%\begin{subequations} \label{eq:CSTR}
\begin{align*}
    \dot{c}_{A,i} &= \dot{V}_{\text{out}}/V_i \ (c_{A,i-1} - c_{A,i}) - r_{1,i} - 2r_{2,i} + u_{A,i}/V_i, \\
    \dot{c}_{B,i} &= \dot{V}_{\text{out}}/V_i \ (c_{B,i-1} - c_{B,i}) - r_{1,i} + u_{B,i}/V_i, \\
    \dot{c}_{R,i} &= \dot{V}_{\text{out}}/V_i \ (c_{R,i-1} - c_{R,i}) + r_{1,i}, \\
    \dot{c}_{S,i} &= \dot{V}_{\text{out}}/V_i \ (c_{S,i-1} - c_{S,i}) + r_{2,i}, \\
    \begin{split}
        \dot{T}_{R,i} &= \dot{V}_{\text{out}}/V_i \ (T_{R,i-1} - T_{R,i}) 
        - r_{1,i} \Delta H_{R1,i}/(\rho c_p) \\
        &\quad - r_{2,i} \Delta H_{R2,i}/(\rho c_p) 
        + kA/(\rho c_p V_i)\ (T_{J,i} - T_{R,i}).
    \end{split}
\end{align*}
%\end{subequations}
The inlet concentrations into the first reactor via $c_{A/B/R/S,0}$ is zero, as the feeding of both educts is done via the inputs $u_{A/B,i}$.
It is assumed that changes in the volume via feeding, temperature and composition changes can be neglected. Therefore, individual reactor volumes $V_i=V_R/n_{R}$ are modeled as constants, as the flow rate $\dot{V}_{\text{out}}$ through the reactors is also constant. The reactor temperatures $T_{R,i}$ are controlled via the jacket temperatures $T_{J,i}$, which are both constrained between 20~$^{\circ}$C and 80~$^{\circ}$C.
The values for the certain and uncertain parameters are given in the supplementary code\footnote{\label{URL} \url{https://github.com/MoritzHein/RobMPCExploitMon}}. The reaction enthalpies $\Delta H_{Rj,i}$ and kinetic constants $k_{j,i}$ are assumed to be uncertain by $\pm 30 \% $. The initial concentrations for all simulations were 0 and the initial temperature was set to be the inlet temperature $T_{R,0}=60~\si{\degreeCelsius}$.
The presented continuous model was discretized via orthogonal collocation with 4 collocation points. %
The control objective is to track an operating point in the cascade where the production of $R$ is maximized by penalizing the offset to the highest stoichiometrically achievable value in each tank with the highest penalty in the last tank. In addition to the constraints on $T_R$, $c_S$ is penalized and constrained to lie below $0.12~\si{\mole \per \cubic\metre}$. Thus Assumption~\ref{ass:Box_constraints} is satisfied. The total respective feeding of $A$ and $B$ are constrained to not exceed $\sum_{i=1}^{n_R} \vu_{A/B,i}\leq 1.5~\si{\mole \per \hour}$, while it is incentivized to feed as much as possible.
The stage-cost function $\ell(\vx_k,\vu_k,\vu_{k-1})$ is formulated quadratically as
\begin{multline} \label{eq:CLC}
    \ell(\vx_k,\vu_k,\vu_{k-1})=(\vx_k-\vx_{\text{set}})^\intercal \mQ(\vx_k-\vx_{\text{set}})\\
    + (\vu_k-\vu_{\text{set}})^{\intercal}\mR(\vu_k-\vu_{\text{set}})+(\vu_k-\vu_{k-1})^{\intercal}\mR_d(\vu_k-\vu_{k-1})\\
    + (\sum u_{A,i,k} - 1.5)^2+ (\sum u_{B,i,k} -1.5)^2.
\end{multline}
The values for the matrices and setpoints in the cost function can be found in the supplementary code\textsuperscript{\ref{URL}}.
The terminal cost is equivalent to 10 times the stage cost without the input terms. The prediction horizon was chosen to be $N= 35$, where one time step corresponds to one hour. The stage cost is calculated for both the top right and bottom left corners and is summed up over the prediction horizon as well as over all subregions.\\
The presented system is not monotone in the states nor the parameters.
However, a continuous decomposition function according to Definition~\ref{def:decomp} can be determined based on the relaxations for the tightest decomposition function from Proposition~\ref{prop:tight_decomp} mentioned in Section~\ref{sec:Mixed_Mon} and is given in the supplementary code\textsuperscript{\ref{URL}}. 
The decomposition function is used in~\eqref{eq:Closed_Loop_Tube} with no additional feedback policy besides the division of the reachable sets. The robust control invariance property of the terminal set was enforced online by replacing~\eqref{eq:Closed_Loop_Tube:subeq:terminal_Constr} with~\eqref{eq:RCIS_relax} guaranteeing the satisfaction of Assumption~\ref{ass:decomp}.
The closed-loop trajectories as well as the closed-loop cost of one closed-loop simulation with $n_R=5$, so 25 states and 20 uncertainties are shown in Figure~\ref{fig:CSTR}. The uncertainties for this simulation were chosen to be the maximum absolute values for reaction rates and reaction enthalpies.
Nominal MPC, in contrast to both robust approaches, leads to unacceptable operation due to temperature constraint violations.
Using the approach $\mathcal{P}_{N}^{\text{CL}_{\vd}}$, leads to lower performance losses compared to $\mathcal{P}_{N}^{\text{OL}_{\vd}}$, which can be seen in the slightly larger end concentration for $\vc_R$, slightly smaller $c_{S,5}$ and a smaller back-off to the temperature constraints. The improvement in accumulated closed-loop cost increases further with time, as it can be seen by the different slopes of the right plots of Figure~\ref{fig:CSTR}. The better performance was achieved because of the introduced recourse using four subregions for the sideproduct concentration $\evc_{S,1}$.

To demonstrate the influence of the number of partitions and the size of the system on the computation time and performance, increasingly many subregions were chosen for decreasingly large systems.
The solutions with one, three and five reactors are presented in Table~\ref{tab:res_CSTR}. The values of the uncertainties for this comparison were chosen to be constant in time, but uniformly distributed random values in each CSTR for the reaction rates and the reaction enthalpies in each of the 50 simulations. The computations were performed on an \emph{Intel i7-3770} CPU. To solve the optimization problems, \emph{CasADi}~\cite{anderssonCasADiSoftwareFramework2019} and \emph{IPOPT}~\cite{wachterImplementationInteriorpointFilter2006a} were used via the simulation and sampling framework of \emph{do-mpc}~\cite{fiedlerDompcFAIRNonlinear2023}.\\%
\begin{table}
    \centering
    \caption{Comparison of the approaches $\mathcal{P}_{N}^{\text{OL}_{\vd}}$ and $\mathcal{P}_{N}^{\text{CL}_{\vd}}$. Average accumulated cost over 75 iterations compared to MPC with full knowledge and average computation time per closed-loop iteration for different $n_R$ over 50 random scenarios.}
    \begin{tabular}{c |c c c c c c c}
         \multirow{2}{*}{$n_R$}& \multirow{2}{*}{$n_x$} & \multirow{2}{*}{$n_p$} & \multirow{2}{*}{$\mu_s$ for $\mathcal{P}_N^{\text{CL}_{\vd}}$} & \multicolumn{2}{c}{Closed-loop cost} & \multicolumn{2}{c}{Comp. time $\left[s\right]$} \\
         & & & &  $\mathcal{P}_{N}^{\text{OL}_{\vd}}$ & $\mathcal{P}_{N}^{\text{CL}_{\vd}}$ & $\mathcal{P}_{N}^{\text{OL}_{\vd}}$ & $\mathcal{P}_{N}^{\text{CL}_{\vd}}$\\ 
         \hline
         1 & 5 & 4 & 32 & 149 \% & 129 \% & 1.70 & 36.40\\
         3 & 15 & 12 & 16 & 161 \% & 114 \% & 16.35 & 111.45 \\ 
         5 & 25 & 20 & 4  & 142 \% & 122 \% & 42.76 &  72.38
    \end{tabular}
    \label{tab:res_CSTR}
    \vspace{-4ex}
\end{table}%
For the case with one tank reactor, recourse was introduced by partitioning the state space once in the dimensions of $T_{R,1},\ c_{R,1}$, three times in $c_{S,1}$ and then once in $c_{A,1}$. For the $n_R=3$ example, the reachable sets were partitioned three times in $c_{A,1}$ and then each subregion three times in $c_{S,1}$. For the $n_R=5$ case, the state $\evc_{S,1}$ was partitioned three times. Which dimension to partition was determined via an a priori parameter study. The number of partitions were chosen to demonstrate the possible tuning of the complexity of the closed-loop approach depending on the problem complexity.
The increase in the closed-loop cost, which is a consequence of robustifying the controllers, can be roughly halfed when using recourse via the division of the reachable sets.\\
For the open-loop approach, the average computation time is just influenced by the system dimension. Nevertheless, a rather large increase can be witnessed going from one reactor to five reactors, which may be due to it being closer to multiple constraints than the closed-loop approach.
The added complexity in the closed-loop approach leads to an increase in average computation time per closed-loop iteration, especially for large $\mu_s$. So for small systems, it was possible to choose large $\mu_s$, while still achieving real time capable sampling times. Nevertheless, the closed-loop approach showed to be scaling well with the system dimension, as it is possible to even control the 5-reactor system with 25 states and 20 uncertainties robustly with not more than twice the computation time, while achieving a better performance compared to the open-loop approach.
\section{Conclusion} \label{sec:Conclusion}
We introduced a robust model predictive control approach that leverages the inherent properties of monotone systems for efficient computation of reachable sets. This method ensures robust constraint satisfaction and recursive feasibility without the computational burden associated with the number of uncertainties. By dividing reachable sets, feedback (or recourse) can be flexibly incorporated into predictions, leading to less conservative control strategies. The concept of mixed-monotonicity broadens our methodology's applicability to non-monotone systems under any feedback law provided that a suitable decomposition function can be found.
The conservatism of over-approximating the reachable set is comparable to existing tube-based approaches.
Validated through a nonlinear high-dimensional chemical tank reactor cascade, our approach demonstrates large potential in handling complex systems with high dimensional uncertainty spaces.

Future work will be based on efficient computation of decomposition functions for general nonlinear systems as well as on the analysis of larger case studies.
\ifCLASSOPTIONcaptionsoff
  \newpage
\fi
\bibliographystyle{IEEEtran}
\bibliography{IEEEabrv,PhD_abr.bib}

% Generated by IEEEtran.bst, version: 1.14 (2015/08/26)
\begin{thebibliography}{10}
\providecommand{\url}[1]{#1}
\csname url@samestyle\endcsname
\providecommand{\newblock}{\relax}
\providecommand{\bibinfo}[2]{#2}
\providecommand{\BIBentrySTDinterwordspacing}{\spaceskip=0pt\relax}
\providecommand{\BIBentryALTinterwordstretchfactor}{4}
\providecommand{\BIBentryALTinterwordspacing}{\spaceskip=\fontdimen2\font plus
\BIBentryALTinterwordstretchfactor\fontdimen3\font minus
  \fontdimen4\font\relax}
\providecommand{\BIBforeignlanguage}[2]{{%
\expandafter\ifx\csname l@#1\endcsname\relax
\typeout{** WARNING: IEEEtran.bst: No hyphenation pattern has been}%
\typeout{** loaded for the language `#1'. Using the pattern for}%
\typeout{** the default language instead.}%
\else
\language=\csname l@#1\endcsname
\fi
#2}}
\providecommand{\BIBdecl}{\relax}
\BIBdecl

\bibitem{rawlingsModelPredictiveControl2017}
J.~B. Rawlings, D.~Q. Mayne, and M.~Diehl, \emph{Model Predictive Control:
  Theory, Computation, and Design}, 2nd~ed.\hskip 1em plus 0.5em minus
  0.4em\relax Madison, Wisconsin: Nob Hill Publishing, 2017.

\bibitem{campoRobustModelPredictive1987}
P.~J. Campo and M.~Morari, ``Robust {{Model Predictive Control}},'' in
  \emph{American {{Control Conference}}}, 1987, pp. 1021--1026.

\bibitem{mayneRobustModelPredictive2005}
D.~Q. Mayne, M.~M. Seron, and S.~V. Rakovi{\'c}, ``Robust model predictive
  control of constrained linear systems with bounded disturbances,''
  \emph{Automatica}, vol.~41, no.~2, pp. 219--224, Feb. 2005.

\bibitem{alamoRobustMPCConstrained2005}
T.~Alamo, D.~Limon, E.~Camacho, and J.~Bravo, ``Robust {{MPC}} of constrained
  nonlinear systems based on interval arithmetic,'' \emph{IEE Proceedings -
  Control Theory and Applications}, vol. 152, no.~3, pp. 325--332, May 2005.

\bibitem{leeRobustRecedingHorizon2000}
Y.~I. Lee and B.~Kouvaritakis, ``Robust receding horizon predictive control for
  systems with uncertain dynamics and input saturation,'' \emph{Automatica},
  vol.~36, no.~10, 2000.

\bibitem{flemingRobustTubeMPC2015a}
J.~Fleming, B.~Kouvaritakis, and M.~Cannon, ``Robust {{Tube MPC}} for {{Linear
  Systems With Multiplicative Uncertainty}},'' \emph{IEEE Trans. Automat.
  Contr.}, vol.~60, no.~4, pp. 1087--1092, Apr. 2015.

\bibitem{houskaRobustOptimizationNonlinear2012}
B.~Houska, F.~Logist, J.~Van~Impe, and M.~Diehl, ``Robust optimization of
  nonlinear dynamic systems with application to a jacketed tubular reactor,''
  \emph{Journal of Process Control}, vol.~22, no.~6, pp. 1152--1160, Jul. 2012.

\bibitem{villanuevaComputingEllipsoidalRobust2017}
M.~E. Villanueva, J.~C. Li, X.~Feng, B.~Chachuat, and B.~Houska, ``Computing
  {{Ellipsoidal Robust Forward Invariant Tubes}} for {{Nonlinear MPC}},''
  \emph{IFAC-PapersOnLine}, vol.~50, no.~1, pp. 7175--7180, Jul. 2017.

\bibitem{fengMinmaxDifferentialInequalities2019}
X.~Feng, H.~Hu, M.~E. Villanueva, and B.~Houska, ``Min-max {{Differential
  Inequalities}} for {{Polytopic Tube MPC}},'' in \emph{2019 {{American Control
  Conference}} ({{ACC}})}, Jul. 2019, pp. 1170--1174.

\bibitem{messererEfficientAlgorithmTubebased2021}
F.~Messerer and M.~Diehl, ``An {{Efficient Algorithm}} for {{Tube-based Robust
  Nonlinear Optimal Control}} with {{Optimal Linear Feedback}},'' in \emph{2021
  60th {{IEEE Conference}} on {{Decision}} and {{Control}} ({{CDC}})}.\hskip
  1em plus 0.5em minus 0.4em\relax Austin, TX, USA: IEEE, Dec. 2021, pp.
  6714--6721.

\bibitem{kouvaritakisModelPredictiveControl2016}
B.~Kouvaritakis and M.~Cannon, \emph{Model {{Predictive Control}}}, ser.
  Advanced {{Textbooks}} in {{Control}} and {{Signal Processing}}.\hskip 1em
  plus 0.5em minus 0.4em\relax Cham: Springer International Publishing, 2016.

\bibitem{mayneTubebasedRobustNonlinear2011}
D.~Q. Mayne, E.~C. Kerrigan, E.~J. {van Wyk}, and P.~Falugi, ``Tube-based
  robust nonlinear model predictive control,'' \emph{International Journal of
  Robust and Nonlinear Control}, vol.~21, no.~11, pp. 1341--1353, 2011.

\bibitem{scokaertMinmaxFeedbackModel1998}
P.~Scokaert and D.~Mayne, ``Min-max feedback model predictive control for
  constrained linear systems,'' \emph{IEEE Trans. Automat. Contr.}, vol.~43,
  no.~8, pp. 1136--1142, Aug. 1998.

\bibitem{delapenaStochasticProgrammingApplied2005}
D.~{de la Pe{\~n}a}, A.~Bemporad, and T.~Alamo, ``Stochastic {{Programming
  Applied}} to {{Model Predictive Control}},'' in \emph{Proceedings of the 44th
  {{IEEE Conference}} on {{Decision}} and {{Control}}}.\hskip 1em plus 0.5em
  minus 0.4em\relax Seville, Spain: IEEE, 2005, pp. 1361--1366.

\bibitem{luciaMultistageNonlinearModel2013c}
S.~Lucia, T.~Finkler, and S.~Engell, ``Multi-stage nonlinear model predictive
  control applied to a semi-batch polymerization reactor under uncertainty,''
  \emph{Journal of Process Control}, vol.~23, no.~9, pp. 1306--1319, Oct. 2013.

\bibitem{subramanianTubeenhancedMultistageMPC2021}
S.~Subramanian, S.~Lucia, R.~Paulen, and S.~Engell, ``Tube-enhanced
  {{Multi-stage MPC}} for {{Flexible Robust Control}} of {{Constrained Linear
  Systems}} with {{Additive}} and {{Parametric Uncertainties}},'' \emph{Int J
  Robust Nonlinear Control}, vol.~31, no.~9, pp. 4458--4487, Jun. 2021.

\bibitem{holtorfMultistageNMPCOnline2019}
F.~Holtorf, A.~Mitsos, and L.~T. Biegler, ``Multistage {{NMPC}} with on-line
  generated scenario trees: {{Application}} to a semi-batch polymerization
  process,'' \emph{Journal of Process Control}, vol.~80, pp. 167--179, Aug.
  2019.

\bibitem{bonzaniniFastApproximateLearningbased2021a}
A.~D. Bonzanini, J.~A. Paulson, G.~Makrygiorgos, and A.~Mesbah, ``Fast
  approximate learning-based multistage nonlinear model predictive control
  using {{Gaussian}} processes and deep neural networks,'' \emph{Computers \&
  Chemical Engineering}, vol. 145, p. 107174, Feb. 2021.

\bibitem{angeliMonotoneControlSystems2003}
D.~Angeli and E.~Sontag, ``Monotone control systems,'' \emph{IEEE Trans.
  Automat. Contr.}, vol.~48, no.~10, pp. 1684--1698, Oct. 2003.

\bibitem{meyerRobustControlledInvariance2016}
P.-J. Meyer, A.~Girard, and E.~Witrant, ``Robust controlled invariance for
  monotone systems: {{Application}} to ventilation regulation in buildings,''
  \emph{Automatica}, vol.~70, pp. 14--20, Aug. 2016.

\bibitem{luciaSetBasedOptimalControl2016}
S.~Lucia, M.~{Schliemann-Bullinger}, R.~Findeisen, and E.~Bullinger, ``A
  {{Set-Based Optimal Control Approach}} for
  {{Pharmacokinetic}}/{{Pharmacodynamic Drug Dosage Design}},''
  \emph{IFAC-PapersOnLine}, vol.~49, no.~7, pp. 797--802, 2016.

\bibitem{cooganMixedMonotonicityReachability2020}
S.~Coogan, ``Mixed {{Monotonicity}} for {{Reachability}} and {{Safety}} in
  {{Dynamical Systems}},'' in \emph{59th {{IEEE Conference}} on {{Decision}}
  and {{Control}}}.\hskip 1em plus 0.5em minus 0.4em\relax Jeju, Korea (South):
  IEEE, Dec. 2020, pp. 5074--5085.

\bibitem{smithDiscreteDynamicsMonotonically2006}
H.~L. Smith, ``The {{Discrete Dynamics}} of {{Monotonically Decomposable
  Maps}},'' \emph{J. Math. Biol.}, vol.~53, no.~4, pp. 747--758, Oct. 2006.

\bibitem{abateRobustlyForwardInvariant2022}
M.~Abate and S.~Coogan, ``Robustly {{Forward Invariant Sets}} for
  {{Mixed-Monotone Systems}},'' \emph{IEEE Trans. Automat. Contr.}, vol.~67,
  no.~9, pp. 4947--4954, Sep. 2022.

\bibitem{harapanahalliContractionGuidedAdaptivePartitioning2023}
A.~Harapanahalli, S.~Jafarpour, and S.~Coogan, ``Contraction-{{Guided Adaptive
  Partitioning}} for {{Reachability Analysis}} of {{Neural Network Controlled
  Systems}},'' in \emph{2023 62nd {{IEEE Conference}} on {{Decision}} and
  {{Control}} ({{CDC}})}, Dec. 2023, pp. 6044--6051.

\bibitem{heinleinRobustMPCApproaches2022}
M.~Heinlein, S.~Subramanian, M.~Molnar, and S.~Lucia, ``Robust {{MPC}}
  approaches for monotone systems*,'' in \emph{2022 {{IEEE}} 61st
  {{Conference}} on {{Decision}} and {{Control}} ({{CDC}})}, Dec. 2022, pp.
  2354--2360.

\bibitem{leeWorstcaseFormulationsModel1997}
J.~Lee and Z.~Yu, ``Worst-case formulations of model predictive control for
  systems with bounded parameters,'' \emph{Automatica}, vol.~33, no.~5, pp.
  763--781, May 1997.

\bibitem{kalliesMonotonicityKineticProofreading2016}
C.~Kallies, M.~Schliemann, R.~Findeisen, S.~Lucia, and E.~Bullinger,
  ``Monotonicity of {{Kinetic Proofreading}},'' \emph{IFAC-PapersOnLine},
  vol.~49, pp. 306--311, Dec. 2016.

\bibitem{abateImprovingFidelityMixedMonotone2021}
M.~Abate and S.~Coogan, ``Improving the {{Fidelity}} of {{Mixed-Monotone
  Reachable Set Approximations}} via {{State Transformations}},'' in
  \emph{American {{Control Conference}}}, May 2021, pp. 4674--4679.

\bibitem{abateVerificationRuntimeAssurance2021}
M.~Abate, M.~Mote, E.~Feron, and S.~Coogan, ``Verification and runtime
  assurance for dynamical systems with uncertainty,'' in \emph{Proceedings of
  the 24th {{International Conference}} on {{Hybrid Systems}}: {{Computation}}
  and {{Control}}}.\hskip 1em plus 0.5em minus 0.4em\relax Nashville Tennessee:
  ACM, May 2021, pp. 1--10.

\bibitem{limonInputtoStateStabilityUnifying2009}
D.~Limon, T.~Alamo, D.~M. Raimondo, D.~M. De~La~Pe{\~n}a, J.~M. Bravo,
  A.~Ferramosca, and E.~F. Camacho, ``Input-to-{{State Stability}}: {{A
  Unifying Framework}} for {{Robust Model Predictive Control}},'' in
  \emph{Nonlinear {{Model Predictive Control}}}, M.~Morari, M.~Thoma, L.~Magni,
  D.~M. Raimondo, and F.~Allg{\"o}wer, Eds.\hskip 1em plus 0.5em minus
  0.4em\relax Berlin, Heidelberg: Springer, 2009, vol. 384, pp. 1--26.

\bibitem{luciaStabilityPropertiesMultistage2020}
S.~Lucia, S.~Subramanian, D.~Limon, and S.~Engell, ``Stability properties of
  multi-stage nonlinear model predictive control,'' \emph{Systems \& Control
  Letters}, vol. 143, p. 104743, Sep. 2020.

\bibitem{bravoRobustMPCConstrained2006}
J.~M. Bravo, T.~Alamo, and E.~F. Camacho, ``Robust {{MPC}} of constrained
  discrete-time nonlinear systems based on approximated reachable sets,''
  \emph{Automatica}, vol.~42, no.~10, pp. 1745--1751, Oct. 2006.

\bibitem{yangTightDecompositionFunctions2019}
L.~Yang and N.~Ozay, ``Tight decomposition functions for mixed monotonicity,''
  in \emph{2019 {{IEEE}} 58th {{Conference}} on {{Decision}} and {{Control}}
  ({{CDC}})}, Dec. 2019, pp. 5318--5322.

\bibitem{kohlerComputationallyEfficientRobust2021}
J.~K{\"o}hler, R.~Soloperto, M.~A. M{\"u}ller, and F.~Allg{\"o}wer, ``A
  computationally efficient robust model predictive control framework for
  uncertain nonlinear systems -- extended version,'' \emph{IEEE Trans. Automat.
  Contr.}, vol.~66, no.~2, pp. 794--801, Feb. 2021.

\bibitem{fiedlerDompcFAIRNonlinear2023}
F.~Fiedler, B.~Karg, L.~L{\"u}ken, D.~Brandner, M.~Heinlein, F.~Brabender, and
  S.~Lucia, ``Do-mpc: {{Towards FAIR}} nonlinear and robust model predictive
  control,'' \emph{Control Engineering Practice}, vol. 140, p. 105676, Nov.
  2023.

\bibitem{anderssonCasADiSoftwareFramework2019}
J.~A.~E. Andersson, J.~Gillis, G.~Horn, J.~B. Rawlings, and M.~Diehl,
  ``{{CasADi}}: A software framework for nonlinear optimization and optimal
  control,'' \emph{Math. Prog. Comp.}, vol.~11, no.~1, pp. 1--36, Mar. 2019.

\bibitem{wachterImplementationInteriorpointFilter2006a}
A.~W{\"a}chter and L.~T. Biegler, ``On the implementation of an interior-point
  filter line-search algorithm for large-scale nonlinear programming,''
  \emph{Math. Program.}, vol. 106, no.~1, pp. 25--57, Mar. 2006.

\end{thebibliography}

\end{document}